\documentclass[10pt]{article}

\usepackage[english]{babel} 
\usepackage{amsmath}                
\usepackage{amsfonts}               
\usepackage{amssymb}                
\usepackage{amsopn}                 
\allowdisplaybreaks[3] 
\usepackage{amsthm}                
\usepackage{bbm}                    
\usepackage{mathrsfs}               
\usepackage[retainorgcmds]{IEEEtrantools} 
\usepackage{calc}                   
\usepackage[dvips]{graphicx}        
\usepackage{epsfig}                
\usepackage{psfrag}                 
\usepackage[dvips]{color}           
\usepackage{fancyhdr}               
\usepackage{verbatim}               
\usepackage{exscale}                
\usepackage{macros}
\newtheorem{theorem}{Theorem}
\newtheorem{lemma}[theorem]{Lemma}
\newtheorem{corollary}[theorem]{Corollary}
\newtheorem{proposition}[theorem]{Proposition}

\title{On Multipath Fading Channels at High SNR}

\author{Tobias Koch ~~ Amos Lapidoth\\\small ETH Zurich\\\small Zurich, Switzerland\\
   \small Email: \{tkoch, lapidoth\}@isi.ee.ethz.ch
 }

\date{}

\begin{document}

\maketitle

\begin{abstract}
  \renewcommand{\thefootnote}{}
  This work studies the capacity of multipath
  fading channels. A noncoherent channel model is considered, where
  neither the transmitter nor the receiver is cognizant of the realization of the
  path gains, but both are cognizant of their statistics. It is
  shown that if the delay spread is large in the sense that the
  variances of the path gains decay exponentially or slower, then
  capacity is bounded in the signal-to-noise ratio (SNR). For such
  channels, capacity does not tend to infinity as the SNR tends to
  infinity. In contrast, if the variances of the path gains decay
  faster than exponentially, then capacity is unbounded in the SNR. It
  is further demonstrated that if the number of paths is finite, then
  at high SNR capacity grows double-logarithmically with the SNR, and the capacity
  pre-loglog, defined as the limiting ratio of capacity to
  $\log\log\SNR$ as $\SNR$ tends to infinity, is $1$ irrespective of
  the number of paths.
  \footnote{The material in this paper
  was presented in part at the 2008 IEEE Information Theory Workshop
  (ITW) Porto, Portugal, at the 2008 IEEE International Symposium on Information
  Theory (ISIT), Toronto, Canada, and at the 2008 IEEE 25-th
  Convention of Electrical and Electronics Engineers in Israel.}
\end{abstract}
\setcounter{footnote}{0}

\section{Introduction}
\label{sec:intro}
We study the capacity of discrete-time \emph{multipath fading
  channels}. In multipath fading channels, the transmitted signal
propagates along a multitude of paths, and the gains and delays of
these paths vary over time. In general, the path delays differ from
each other, and the receiver thus observes a weighted sum of delayed
replicas of the transmitted signal, where the weights are random. We
  shall slightly abuse nomenclature and refer
to each summand in the received signal as a path, and to the
corresponding weight as its path gain, even if it is in fact
composed of a multitude of paths. We consider a \emph{noncoherent} channel model, where
transmitter and receiver are cognizant of the statistics of the path
gains, but are ignorant of their realization.

Multipath fading channels arise in wireless
communications, where obstacles in the surroundings reflect the
transmitted signal and force it to propagate along multiple paths, and where
relative movements of transmitter, receiver, and obstacles lead to
time-variations of the path gains and delays.
Examples of wireless communication scenarios where the receiver
observes typically more than one path
include \emph{radio communications} (particularly if the transmitted
signal is of large bandwidth as, for example, in \emph{Ultra-Wideband} or
in \emph{CDMA}) and \emph{underwater acoustic communications}.

The capacity of noncoherent multipath fading channels has been
investigated extensively in the wideband regime, where the
signal-to-noise ratio (SNR) is typically small. It was shown by Kennedy that,
in the limit as the available bandwidth tends to infinity, the
capacity of the fading channel is the same as the capacity of the
additive white Gaussian noise (AWGN) channel of equal received power;
see \cite[Sec.~8.6]{gallager68} and references therein.

To the best of our knowledge, not much is known about the capacity of
noncoherent multipath fading channels at high SNR. For the special case of
noncoherent \emph{frequency-flat} fading channels (where we only have
\emph{one} path), it was shown by Lapidoth \& Moser
\cite{lapidothmoser03_3} that if the fading process is of finite
entropy rate, then at high SNR capacity grows double-logarithmically
in the SNR. This is much slower than the logarithmic growth of the
AWGN capacity \cite{shannon48}.

In this work, we study the high-SNR behavior of the capacity of
noncoherent \emph{multipath} fading channels (where the number of paths is typically
greater than one). We demonstrate that the capacity of such channels does not
merely grow slower with the SNR than the capacity of the AWGN
channel, but may be even \emph{bounded} in the SNR. In other words,
for such channels the capacity does not necessarily tend to infinity
as the SNR tends to infinity.

We derive a necessary and a sufficient condition for the
capacity to be bounded in the SNR. We show that if the variances of
the path gains decay \emph{exponentially or slower}, then capacity is
bounded in the SNR. In contrast, if the variances of the path
gains decay \emph{faster than exponentially}, then capacity is
unbounded in the SNR. We further show that if the number of paths is
finite, then at high SNR capacity increases double-logarithmically with the SNR, and
the capacity pre-loglog, defined as the limiting ratio of
the capacity to $\log\log\SNR$ as $\SNR$ tends to infinity, is $1$
irrespective of the number of paths.

The rest of this paper is organized as follows. We begin with a
mathematical description of the considered channel model in
Section~\ref{sec:channelmodel}. Section~\ref{sec:capacity} is devoted
to channel capacity. Our main results are summarized in
Section~\ref{sec:result}. They follow from upper bounds and
lower bounds on channel capacity, which are derived in
Sections~\ref{sec:upperproof}  and \ref{sec:lowerproof},
respectively. Section~\ref{sec:summary} concludes the paper with a
brief summary and a discussion of our results.

\section{Channel Model}
\label{sec:channelmodel}
Let $\Complex$ and $\Naturals$ denote the set of complex numbers and
the set of positive integers, respectively. We consider a
discrete-time multipath fading channel whose channel output
$Y_k\in\Complex$ at time $k\in\Naturals$ corresponding to the time-1
through time-$k$ channel inputs $x_1,\ldots,x_k\in\Complex$ is given
by
\begin{equation}
  Y_k = \sum_{\ell=0}^{k-1} H_k^{(\ell)} x_{k-\ell}+Z_k, \quad
  k\in\Naturals.
\end{equation}
Here $\{Z_k\}$ models additive noise, and $H_k^{(\ell)}$ denotes the
time-$k$ gain of the $\ell$-th path. We assume that $\{Z_k\}$ is
a sequence of independent and identically distributed (IID),
zero-mean, variance-$\sigma^2$, circularly-symmetric, complex Gaussian
random variables. For each path $\ell\in\Naturals_0$ (where
$\Naturals_0$ denotes the set of nonnegative integers), we assume that
$\bigl\{H_k^{(\ell)},\;k\in\Naturals\bigr\}$ is a zero-mean, complex stationary
process. We denote its variance and its differential entropy rate by
\begin{equation}
  \alpha_{\ell} \triangleq \E{\bigl|H_k^{(\ell)}\bigr|^2}, \qquad \ell\in\Naturals_0
\end{equation}
and
\begin{equation}
  h_{\ell} \triangleq \lim_{n\to\infty} \frac{1}{n}
  h\bigl(H_1^{(\ell)},\ldots,H_n^{(\ell)}\bigr), \quad \ell\in\Naturals_0.
\end{equation}
We shall say that the channel has a \emph{finite number of paths}, if for
some finite integer $\const{L}\in\Naturals_0$ 
\begin{equation}
  H_k^{(\ell)}=0, \quad\ell>\const{L}, \quad k\in\Naturals.
\end{equation}
We assume that $\alpha_0>0$. We further assume
\begin{equation}
  \label{eq:alpha}
  \sup_{\ell\in\Naturals_0} \alpha_{\ell} <\infty
\end{equation}
and
\begin{equation}
  \inf_{\ell\in\set{L}} h_{\ell} > -\infty,\label{eq:finiteentropy}
\end{equation}
where the set $\set{L}$ is defined as $\set{L}\triangleq
\{\ell\in\Naturals_0:\alpha_{\ell}>0\}$. (When the path gains are
Gaussian, then the latter condition \eqref{eq:finiteentropy} is equivalent to saying that the
mean-square error in predicting the present path gain from its past is
strictly positive, i.e., that the present path gain cannot be predicted
perfectly from its past.) We finally assume that the processes
\begin{equation*}
  \bigl\{H_k^{(0)},\;k\in\Naturals\bigr\},\bigl\{H_k^{(1)},\;k\in\Naturals\bigr\},\ldots
\end{equation*}
are independent (``uncorrelated scattering''); that they are jointly
independent of $\{Z_k\}$; and that the joint law of
\begin{equation*}
  \left(\{Z_k\},\bigl\{H_k^{(0)},\;k\in\Naturals\bigr\},\bigl\{H_k^{(1)},\;k\in\Naturals\bigr\},\ldots\right)
\end{equation*}
does not depend on the input sequence $\{x_k\}$. We consider a
noncoherent channel model where neither transmitter nor receiver is
cognizant of the realization of
$\bigl\{H_k^{(\ell)},\;k\in\Naturals\bigr\}$, $\ell\in\Naturals_0$,
but both are aware of their law. We do not assume that the path
gains are Gaussian.

\section{Channel Capacity}
\label{sec:capacity}
Let $A_m^n$ denote the sequence $A_m,\ldots,A_n$. We define the
capacity (in nats per channel use) as
\begin{equation}
  \label{eq:capacity}
  C(\SNR) \triangleq \varliminf_{n\to\infty}\frac{1}{n} \sup I\bigl(X_1^n;Y_1^n\bigr),
\end{equation}
where the supremum is over all joint distributions on $X_1,\ldots,X_n$
satisfying the power constraint
\begin{equation}
  \label{eq:power}
  \frac{1}{n} \sum_{k=1}^n \E{|X_k|^2} \leq \const{P},
\end{equation}
and where SNR is defined as
\begin{equation}
  \SNR \triangleq \frac{\const{P}}{\sigma^2}.
\end{equation}

By Fano's inequality, no rate above $C(\SNR)$ is achievable. (See
\cite{coverthomas91} for a definition of an achievable rate.) We do
not claim that there is a coding theorem associated with
\eqref{eq:capacity}, i.e., that $C(\SNR)$ is achievable. A coding
theorem will hold, for example, if the number of paths is finite,
and if the processes corresponding to these paths
$\bigl\{H_k^{(0)},\;k\in\Naturals\bigr\},\ldots,\bigl\{H_k^{(\const{L})},\;k\in\Naturals\bigr\}$
are jointly ergodic, see \cite[Thm.~2]{kim08}.

The special case of noncoherent frequency-flat fading channels
(where we have only one path)
was studied by Lapidoth and Moser \cite{lapidothmoser03_3}.
They showed that if the fading process
$\bigl\{H_k^{(0)},\;k\in\Naturals\bigr\}$ is ergodic, then the
capacity satisfies
\begin{equation}
  \label{eq:lapidothmoser}
  \lim_{\SNR\to\infty} \bigl\{C(\SNR) -\log\log\SNR \bigr\} = \log\pi
  +\E{\log\bigl|H_1^{(0)}\bigr|^2} - h_0
\end{equation}
(see \cite[Thm.~4.41]{lapidothmoser03_3}), where $\log(\cdot)$ denotes
the natural logarithm function. Thus, at high SNR, the
capacity of noncoherent frequency-flat fading channels grows
double-logarithmically with the SNR. Lapidoth and Moser concluded that
communicating over noncoherent frequency-flat fading channels at high
SNR is extremely power-inefficient, as one should expect to square the
SNR for every additional bit per channel use.\footnote{Note that the capacity of coherent
  fading channels (where the fading realization is known to the
  receiver) behaves logarithmically with the SNR
  \cite{ericson70}. Thus in the coherent case it suffices to double
  the SNR for every additional bit per channel use.}

In this paper, we show \emph{inter alia} that communicating over noncoherent
multipath fading channels at high SNR is not merely power-inefficient,
but may be even worse: if the delay spread is large in the sense that
the sequence $\{\alpha_{\ell}\}$ (which describes the variances of the
path gains) decays exponentially or slower,
then capacity is bounded in the SNR. For such channels,
capacity does not tend to infinity as the SNR tends to infinity. The
main results of this paper are presented in the following section.

\section{Main Results}
\label{sec:result}

Our main results are a sufficient and a necessary condition
on $\{\alpha_{\ell}\}$ for $C(\SNR)$ to be bounded in
$\SNR$, as well as a characterization of the capacity pre-loglog when
the number of paths is finite.

\begin{theorem}
  \label{thm:bounded}
  Consider the above channel model. Then
  \begin{IEEEeqnarray}{rcCl}
    (i) & \quad\biggl(\varliminf_{\ell\to\infty}
    \frac{\alpha_{\ell+1}}{\alpha_{\ell}}>0\biggr)\quad & \Longrightarrow &
    \quad\biggl(\,\sup_{\SNR>0} C(\SNR)<\infty\biggr)\\
    (ii) & \quad \biggl(\lim_{\ell\to\infty}\frac{1}{\ell}\log\frac{1}{\alpha_{\ell}}=\infty\biggr)\quad
    & \Longrightarrow & \quad\biggl(\,\sup_{\SNR>0} C(\SNR) = \infty\biggr),
  \end{IEEEeqnarray}
  where we define $a/0\triangleq \infty$ for every $a>0$ and
  $0/0\triangleq 0$.
\end{theorem}
\begin{proof}
  Part~(i) is proven in Section~\ref{sub:upper1}, and Part~(ii) is
  proven in Sections~\ref{sub:lower} \& \ref{sub:neccondbounded}.
\end{proof}

By noting that
\begin{equation*}
  \biggl(\lim_{\ell\to\infty}
  \frac{\alpha_{\ell+1}}{\alpha_{\ell}}=0\biggr) \quad \Longrightarrow
  \quad \biggl(\lim_{\ell\to\infty}\frac{1}{\ell}\log\frac{1}{\alpha_{\ell}}=0\biggr)
\end{equation*}
we obtain from Theorem~\ref{thm:bounded} the immediate corollary:

\begin{corollary}
  \label{cor:bounded}
  Consider the above channel model. Then
  \begin{IEEEeqnarray}{rcCl}
    (i) & \quad\biggl(\varliminf_{\ell\to\infty}
    \frac{\alpha_{\ell+1}}{\alpha_{\ell}}>0\biggr)\quad & \Longrightarrow &
    \quad\biggl(\,\sup_{\SNR>0} C(\SNR)<\infty\biggr)\label{eq:cor1}\\
    (ii) & \quad \biggl(\lim_{\ell\to\infty} \frac{\alpha_{\ell+1}}{\alpha_{\ell}}=0\biggr)\quad
    & \Longrightarrow & \quad\biggl(\,\sup_{\SNR>0} C(\SNR) = \infty\biggr),\label{eq:cor2}
  \end{IEEEeqnarray}
  where we define $a/0\triangleq \infty$ for every $a>0$ and
  $0/0\triangleq 0$.
\end{corollary}
For example, if
\begin{equation}
  \alpha_{\ell} = e^{-\ell}, \quad \ell\in\Naturals_0,
\end{equation}
then
\begin{equation}
  \lim_{\ell\to\infty} \frac{\alpha_{\ell+1}}{\alpha_{\ell}} = \frac{1}{e} >0
\end{equation}
and it follows from Part~(i) of Corollary~\ref{cor:bounded} that the capacity is bounded
in the SNR. On the other hand, if
\begin{equation}
  \alpha_{\ell} = \exp\bigl(-\ell^{\kappa}\bigr), \quad \ell\in\Naturals_0
\end{equation}
for some $\kappa>1$, then
\begin{equation}
  \lim_{\ell\to\infty} \frac{\alpha_{\ell+1}}{\alpha_{\ell}} =
  \lim_{\ell\to\infty} \exp\bigl(\ell^{\kappa}-(\ell+1)^{\kappa}\bigr)
  = 0
\end{equation}
and it follows from Part~(ii) of Corollary~\ref{cor:bounded} that the
capacity is unbounded in the SNR. Roughly speaking, we can say that when $\{\alpha_{\ell}\}$
decays \emph{exponentially or slower}, then $C(\SNR)$ is bounded in
$\SNR$, and when $\{\alpha_{\ell}\}$ decays \emph{faster than exponentially},
then $C(\SNR)$ is unbounded in $\SNR$.

The condition on the left-hand side (LHS) of \eqref{eq:cor2} is surely
satisfied if the channel has a finite number of paths, as in this case
\begin{equation*}
  H_k^{(\ell)}=0,\quad\ell>\const{L}, \quad k\in\Naturals,
\end{equation*}
which implies
\begin{equation*}
  \alpha_{\ell}=0,\quad \ell>\const{L} \qquad \textnormal{and}\qquad
  \frac{\alpha_{\ell+1}}{\alpha_{\ell}} = \frac{0}{0} \triangleq 0,
  \quad \ell>\const{L}.
\end{equation*}
Consequently, it follows from Corollary~\ref{cor:bounded} that if the
number of paths is finite, then
$C(\SNR)$ is unbounded in $\SNR$. However, for this case the high-SNR
behavior of the capacity can be characterized more accurately:
Theorem~\ref{thm:finite} ahead shows that if the number of paths is
finite, then the capacity pre-loglog, defined as
\begin{equation}
  \Lambda \triangleq \varlimsup_{\SNR\to\infty} \frac{C(\SNR)}{\log\log\SNR},
\end{equation}
is $1$ irrespective of the number of paths. The pre-loglog in this
case is thus the same as for frequency-flat fading.

\begin{theorem}
  \label{thm:finite}
  Consider the above channel model. Further assume that the number of
  paths is finite. Then, irrespective of the number of paths, the
  capacity pre-loglog is given by
  \begin{equation}
    \Lambda=\lim_{\SNR\to\infty}\frac{C(\SNR)}{\log\log\SNR} = 1.
  \end{equation}
\end{theorem}
\begin{proof}
  See Section~\ref{sub:upper2} for the converse and
  Sections~\ref{sub:lower} \& \ref{sub:preloglog} for the direct
  part.
\end{proof}

When studying multipath fading channels at low or at moderate
SNR, it is often assumed that the channel has a finite number
of paths, even if the number of paths is in reality infinite. This
assumption is commonly justified by saying that only the first $(\const{L}+1)$
paths are relevant, since the variances of the remaining paths are
typically small and hence the influence of these paths on the capacity
is marginal. As we see from Theorems~\ref{thm:bounded} \&
\ref{thm:finite}, this argument is not valid anymore when studying
multipath fading channels at high SNR. In fact, when for example the
sequence of variances $\{\alpha_{\ell}\}$ decays exponentially, then
according to Part~(i) of Theorem~\ref{thm:bounded} the capacity is
bounded in the SNR. However, if we consider only the first $(\const{L}+1)$
paths and set the other paths to zero, then it follows from
Theorem~\ref{thm:finite} that, irrespective of $\const{L}$, the
capacity increases double-logarithmically with the SNR. Thus, even
though the variances of the remaining paths $\alpha_{\ell}$,
$\ell>\const{L}$ can be made arbitrarily small by choosing $\const{L}$
sufficiently large, these paths may have a significant influence on the
capacity behavior at high SNR.

The reason why paths with a small variance can affect the
capacity behavior is that the capacity depends on the variance of the
product between the path gains and the transmitted signal and not on
the variance of the path gains only. Since at high SNR the variance
of $\sum_{\ell=\const{L}+1}^{\infty} H_k^{(\ell)}X_{k-\ell}$ might be
huge even if the variance of $\sum_{\ell=\const{L}+1}^{\infty}
H_k^{(\ell)}$ is small, the relevance of a path is determined not only
by its own variance but also by the power available at the
transmitter.
The number of paths that are needed to approximate a multipath
channel typically depends on the $\SNR$ and may grow to infinity as the
$\SNR$ tends to infinity.

In order to prove the above results, we derive upper and lower bounds on
the capacity. Since these bounds may also be of independent
interest, we summarize them in the following propositions.

\begin{proposition}[Upper Bounds] $\quad$
  \label{prop:upperbounds}
  \begin{enumeratePart}
  \item\label{prop:upperbound1}
    Consider the above channel model. Further assume that for some $0<\rho<1$
    and some $\ell_0\in\Naturals$
    \begin{equation*}
      \alpha_{\ell_0}>0 \qquad \textnormal{and} \qquad \frac{\alpha_{\ell+1}}{\alpha_{\ell}} \geq \rho, \quad \ell\geq\ell_0.
    \end{equation*}
    Then the capacity $C(\SNR)$ is upper bounded by
    \begin{equation}
      \label{eq:propupperbound1}
      C(\SNR) \leq \log\frac{2\pi^2}{\sqrt{\tilde{\rho}}} -
      \inf_{\ell\in\set{L}}(h_{\ell}-\log\alpha_{\ell}),\quad \SNR\geq 0,
    \end{equation}
    where
    \begin{equation}
      \tilde{\rho} =
      \min\Bigl\{\rho^{\ell_0-1}\frac{\alpha_{\ell_0}}{\max_{0\leq
          \ell'<\ell_0} \alpha_{\ell'}},\rho^{\ell_0}\Bigr\}.
    \end{equation}
  \item\label{prop:upperbound2}
    Consider the above channel model. Further assume that
    \begin{equation}
      \label{eq:alphasum}
      \sum_{\ell=0}^{\infty} \alpha_{\ell} \triangleq\alpha < \infty.
    \end{equation}
    Then
    \begin{equation}
      \varlimsup_{\SNR\to\infty} \bigl\{C(\SNR) -
      \log\log\SNR\big\} \leq 1+\log\pi - \inf_{\ell\in\set{L}}(h_{\ell}-\log\alpha_{\ell}).
    \end{equation}
  \end{enumeratePart}
  \end{proposition}
  \begin{proof}
    Part~\ref{prop:upperbound1} is proven in Section~\ref{sub:upper1},
    and Part~\ref{prop:upperbound2} in Section~\ref{sub:upper2}.
  \end{proof}
  For example, if $\{\alpha_{\ell}\}$ is a geometric sequence, i.e.,
  \begin{equation*}
    \alpha_{\ell} = \rho^{\ell}, \quad \ell\in\Naturals_0
  \end{equation*}
  for some $0<\rho<1$, and if the path gains are Gaussian and
  memoryless so
  \begin{equation*}
    h_{\ell} = \log(\pi e \alpha_{\ell}), \quad\ell\in\Naturals_0,
  \end{equation*}
  then Part~(i) of Proposition~\ref{prop:upperbounds} yields
  \begin{equation}
    C(\SNR) \leq \log\frac{2\pi}{\sqrt{\rho}}-1, \quad \SNR\geq 0.
  \end{equation}
  Part~(ii) of Proposition~\ref{prop:upperbounds} combines with
  \eqref{eq:lapidothmoser} to show that the
  pre-loglog of a multipath fading channel can never be larger than
  the pre-loglog of a frequency-flat fading channel. This result is
  consistent with the intuition that at high SNR the multipath
  behavior is detrimental.
  
  Our last result is a lower bound on the capacity. This bound
  is the basis for the proof of Part~(ii) of Theorem~\ref{thm:bounded} and
  for the direct part of Theorem~\ref{thm:finite}. 
  \begin{proposition}[Lower Bound]
    \label{prop:lowerbound}
    Consider the above channel model. Further assume that
    \begin{equation}
      \label{eq:alphasumlower}
      \sum_{\ell=0}^{\infty} \alpha_{\ell} \triangleq \alpha < \infty.
    \end{equation}
    Let $\const{L}(\const{P})\in\Naturals$ be some positive integer
    that satisfies
    \begin{equation}
      \label{eq:Lcondition}
      \sum_{\ell=\const{L}(\const{P})+1}^{\infty}\alpha_{\ell}\,
      \const{P} \leq \sigma^2
    \end{equation}
    (typically $\const{L}(\const{P})$ depends on $\const{P}$),
    and let $\tau\in\Naturals$ be some arbitrary positive integer that is
    allowed to depend on $\const{L}(\const{P})$.
    Then the capacity $C(\SNR)$ is lower bounded by
    \begin{IEEEeqnarray}{lCl}
      C(\SNR) \;& \geq & \frac{\tau}{\const{L}(\const{P})+\tau}
      \log\log\const{P}^{1/\tau}+\frac{\tau}{\const{L}(\const{P})+\tau}\biggl(\E{\log\bigl|H_1^{(0)}\bigr|^2}-1-2\log\Bigl(\sqrt{\alpha_0}+\sqrt{\alpha+2\sigma^2}\Bigr)\biggr),\quad\nonumber\\
      \IEEEeqnarraymulticol{3}{r}{\const{P}>1.\IEEEeqnarraynumspace}\label{eq:proplowerbound}
    \end{IEEEeqnarray}
  \end{proposition}
  \begin{proof}
    See Section~\ref{sub:lower}.
  \end{proof}
  
\section{Proofs of the Upper Bounds}
\label{sec:upperproof}

In this section, we establish a proof of
Proposition~\ref{prop:upperbounds}, which in turn will be used to
prove Part~(i) of Theorem~\ref{thm:bounded} and the converse to
Theorem~\ref{thm:finite}.

Part~(i) of Proposition~\ref{prop:upperbounds} is proven in
Section~\ref{sub:upper1}, and it is demonstrated that Part~(i) of
Theorem~\ref{thm:bounded} follows immediately from this
result. Section~\ref{sub:upper2} proves Part~(ii) of
Proposition~\ref{prop:upperbounds}. This part provides an upper bound
on the capacity pre-loglog and will be used later, together with a
capacity lower bound that is derived in Section~\ref{sec:lowerproof}, to
establish Theorem~\ref{thm:finite}.

\subsection{Bounded Capacity}
\label{sub:upper1}
We provide a proof of Part~(i) of
Proposition~\ref{prop:upperbounds} by deriving an upper bound on
channel capacity that holds under the assumption that for some
$0<\rho<1$ and some $\ell_0\in\Naturals_0$
\begin{equation}
  \label{eq:boundedassumption}
  \alpha_{\ell_0} > 0 \qquad \textnormal{and} \qquad
  \frac{\alpha_{\ell+1}}{\alpha_{\ell}} \geq \rho, \quad \ell\geq\ell_0.
\end{equation}
As this bound is finite for $\SNR\geq 0$, Part~(i) of
Theorem~\ref{thm:bounded} follows immediately from
Part~(i) of Proposition~\ref{prop:upperbounds} by noting that if
\begin{equation*}
  \varliminf_{\ell\to\infty} \frac{\alpha_{\ell+1}}{\alpha_{\ell}}>0,
\end{equation*}
then we can find a $0<\rho<1$ and an $\ell_0\in\Naturals$ satisfying
\eqref{eq:boundedassumption}.

The proof of the desired upper bound is akin to the proof of an upper
bound that was derived in
\cite[Sec.~6.1]{kochlapidothsotiriadis08_1_submitted_to}. (However,
\cite{kochlapidothsotiriadis08_1_submitted_to} studies a channel whose
inputs \& outputs take value in the set of real numbers rather than in
$\Complex$.)
It is based on \eqref{eq:capacity} and on
an upper bound on $\frac{1}{n} I(X_1^n;Y_1^n)$. To this end, we begin
with the chain rule for mutual information \cite[Thm.~2.5.2]{coverthomas91}
\begin{IEEEeqnarray}{lCl}
  \frac{1}{n} I(X_1^n;Y_1^n)
  & = & \frac{1}{n} \sum_{k=1}^{\ell_0}
  I\big(X_1^n;Y_k\big|Y_1^{k-1}\big) + \frac{1}{n}\sum_{k=\ell_0+1}^n
  I\big(X_1^n;Y_k\big|Y_1^{k-1}\big).\IEEEeqnarraynumspace \label{eq:chainrule}
\end{IEEEeqnarray}
Each term in the first sum on the right-hand side (RHS) of
\eqref{eq:chainrule} is upper bounded by
\begin{IEEEeqnarray}{lCl}
  I\big(X_1^n;Y_k\big|Y_1^{k-1}\big) & \leq & h(Y_k) -
  h\Big(Y_k\Big|Y_1^{k-1},X_1^n,H_k^{(0)},H_k^{(1)},\ldots,H_k^{(k-1)}\Big)\nonumber\\
  & \leq & \log\left(\pi
  e\left(\sigma^2+\sum_{\ell=0}^{k-1} \alpha_{\ell}\E{|X_{k-\ell}|^2}\right)\right)
  - \log\big(\pi e\sigma^2\big) \nonumber\\
  & \leq & \log\left(1+\sup_{\ell\in\Naturals_0} \alpha_{\ell} \, n
  \, \SNR\right), \label{eq:firstl}
\end{IEEEeqnarray}
where the first inequality follows because conditioning cannot increase
differential entropy \cite[Thm.~9.6.1]{coverthomas91}; the second inequality follows from the entropy maximizing
property of Gaussian random variables \cite[Thm.~9.6.5]{coverthomas91}; and the last inequality
follows by upper bounding $\alpha_{\ell}\leq
\sup_{\ell'\in\Naturals_0}\alpha_{\ell'}$, $\ell=0,1,\ldots,k-1$
and from the power constraint \eqref{eq:power}.

For $k=\ell_0+1,\ell_0+2,\ldots,n$, we upper bound
$I\big(X_1^n;Y_k\big|Y_1^{k-1}\big)$ using the general upper bound for
mutual information \cite[Thm.~5.1]{lapidothmoser03_3}
\begin{equation}
  I(X;Y) \leq \int D\big(W(\cdot|x)\big\| R(\cdot)\big) \d Q(x), \label{eq:duality}
\end{equation}
where $D(\cdot\|\cdot)$ denotes relative entropy, i.e.,
\begin{equation*}
  D(P_1\|P_0) = \left\{\begin{array}{ll}\displaystyle \int
  \log\frac{\d P_1}{\d P_0}\d P_1 \quad & \textnormal{if }P_1 \ll P_0 \\
  +\infty & \textnormal{otherwise,}\end{array}\right.
\end{equation*}
$W(\cdot|\cdot)$ is
the channel law, $Q(\cdot)$ denotes
the distribution on the channel input $X$, and $R(\cdot)$ is any
distribution on the output alphabet.\footnote{For channels with
  finite input and output alphabets this inequality follows by
  Tops{\o}e's identity \cite{topsoe67}; see also
  \cite[Thm.~3.4]{csiszarkorner81}.} Thus any choice of output
distribution $R(\cdot)$ yields an upper bound on the mutual
information.

For any given $Y_1^{k-1}=y_1^{k-1}$, we choose the output distribution
  $R(\cdot)$ to be of density
\begin{equation}
  \frac{\sqrt{\beta}}{\pi^2 |y_k|}\frac{1}{1+\beta
  |y_k|^2}, \qquad y_k \in \Complex, \label{eq:cauchy}
\end{equation}
with $\beta=1/(\tilde{\rho}|y_{k-\ell_0}|^2)$ and
\begin{equation}
  \tilde{\rho} = \min\left\{\rho^{\ell_0-1}
  \frac{\alpha_{\ell_0}}{\max_{0\leq\ell'<\ell_0}
  \alpha_{\ell'}},\rho^{\ell_0}\right\}.
\end{equation}
(If $y_{k-\ell_0}=0$, then the density \eqref{eq:cauchy} is
  undefined. However, this event is of zero probability and has
  therefore no impact on the mutual information $I\big(X_1^n;Y_k\big|Y_1^{k-1}\big)$.)
With this choice
\begin{equation}
  0 < \tilde{\rho} < 1 \qquad \textnormal{and} \qquad \tilde{\rho} \, \alpha_{\ell}
  \leq \alpha_{\ell+\ell_0}, \quad \ell \in \Naturals_0. \label{eq:beta}
\end{equation}
Using \eqref{eq:cauchy} in \eqref{eq:duality}, and averaging over
$Y_1^{k-1}$, we obtain
\begin{IEEEeqnarray}{lCl}
  I\big(X_1^n;Y_k\big|Y_1^{k-1}\big) & \leq & \frac{1}{2} \E{\log
    |Y_k|^2} + \frac{1}{2}\E{\log\big(\tilde{\rho}
    |Y_{k-\ell_0}|^2\big)} +
  \E{\log\biggl(1+\frac{|Y_k|^2}{\tilde{\rho}|Y_{k-\ell_0}|^2}\biggr)}\nonumber\\
  & & {} - h\big(Y_k\big|X_1^n,Y_1^{k-1}\big) + \log \pi^2\nonumber\\
  & = & \frac{1}{2} \E{\log|Y_k|^2} -
  \frac{1}{2}\E{\log|Y_{k-\ell_0}|^2}+\E{\log\bigl(\tilde{\rho}|Y_{k-\ell_0}|^2+|Y_k|^2\bigr)}\nonumber\\
  & & {}  - h\big(Y_k\big|X_1^n,Y_1^{k-1}\big) + \log \frac{\pi^2}{\sqrt{\tilde{\rho}}}.\label{eq:up1}
\end{IEEEeqnarray}

We bound the third and the fourth term in \eqref{eq:up1}
separately. We begin with
\begin{IEEEeqnarray}{lCl}
  \E{\log\big(\tilde{\rho}|Y_{k-\ell_0}|^2+|Y_k|^2\big)} & = &
  \E{\Econd{\log\big(\tilde{\rho}|Y_{k-\ell_0}|^2+|Y_k|^2\big)}{X_1^k}}\nonumber\\
  & \leq &
  \E{\log\Bigl(\tilde{\rho}\Econd{|Y_{k-\ell_0}|^2}{X_1^k}+\Econd{|Y_k|^2}{X_1^k}\Bigr)}\nonumber\\
  & = &
  \E{\log\biggl((1+\tilde{\rho})\sigma^2+\sum_{\ell=0}^{k-\ell_0-1}\tilde{\rho}\,\alpha_{\ell}|X_{k-\ell_0-\ell}|^2
  + \sum_{\ell=0}^{k-1}\alpha_{\ell} |X_{k-\ell}|^2\biggr)}\nonumber\\
  & \leq & \E{\log\left(2\sigma^2+\sum_{\ell=0}^{k-\ell_0-1}
      \alpha_{\ell+\ell_0}|X_{k-\ell_0-\ell}|^2+\sum_{\ell=0}^{k-1}
      \alpha_{\ell}|X_{k-\ell}|^2\right)}\nonumber\\
  & = & \E{\log\left(2\sigma^2+\sum_{\ell'=\ell_0}^{k-1}
      \alpha_{\ell'}|X_{k-\ell'}|^2+\sum_{\ell=0}^{k-1}
      \alpha_{\ell}|X_{k-\ell}|^2\right)}\nonumber\\
  & \leq & \log 2 + \E{\log\left(\sigma^2+\sum_{\ell=0}^{k-1} \alpha_{\ell}|X_{k-\ell}|^2\right)},\label{eq:U3}
\end{IEEEeqnarray}
where the first inequality follows by Jensen's inequality; the
subsequent equality follows by evaluating the expectations; the next
inequality by \eqref{eq:beta}; the following equality by substituting
$\ell'=\ell+\ell_0$; and the last inequality follows because
\begin{equation*}
  \sum_{\ell=\ell_0}^{k-1}\alpha_{\ell}|X_{k-\ell}|^2 \leq
  \sum_{\ell=0}^{k-1} \alpha_{\ell} |X_{k-\ell}|^2.
\end{equation*}
Next we derive a lower bound on
$h\big(Y_k\big|X_1^n,Y_1^{k-1}\big)$. Let 
\begin{equation}
  \Bigl\{H_{k'}^{(\ell)}\Bigr\}_{k'=1}^{k-1} =
  \Bigl(H_1^{(\ell)},H_2^{(\ell)},\ldots,H_{k-1}^{(\ell)}\Bigr), \quad \ell\in\Naturals_0,
\end{equation}
and let
\begin{equation}
  \vect{H}_1^{k-1} =
  \biggl(\Bigl\{H_{k'}^{(0)}\Bigr\}_{k'=1}^{k-1},\Bigl\{H_{k'}^{(1)}\Bigr\}_{k'=1}^{k-1},\ldots,\Bigl\{H_{k'}^{(k-1)}\Bigr\}_{k'=1}^{k-1}\biggr).
\end{equation}
We have
\begin{IEEEeqnarray}{lCl}
  h\big(Y_k\big|X_1^n,Y_1^{k-1}\big) & \geq &
  h\big(Y_k\big|X_1^n,Y_1^{k-1},\vect{H}_{1}^{k-1}\big)\nonumber\\
  & = & h\big(Y_k\Big|X_1^n,\vect{H}_{1}^{k-1}\big),
\end{IEEEeqnarray}
where the inequality follows because conditioning cannot increase
differential entropy; and where the equality follows because, conditional on
$\big(X_1^n,\vect{H}_{1}^{k-1}\big)$, $Y_k$ is independent of
$Y_1^{k-1}$.
Let $\set{S}_k$ be defined as
\begin{equation}
  \label{eq:S}
  \set{S}_k\triangleq\{\ell = 0,1,\ldots,k-1:
  |x_{k-\ell}|^2\,\alpha_{\ell}>0\}.
\end{equation}
Using the entropy power inequality \cite[Thm.~16.6.3]{coverthomas91},
and using that the processes
\begin{equation*}
\big\{H_{k}^{(0)}, k\in\Naturals\big\},\big\{H_{k}^{(1)}, k\in\Naturals\big\},\ldots
\end{equation*}
are independent and jointly independent of $X_1^n$, it is shown in Appendix~\ref{app:EPI} that for any given $X_1^n=x_1^n$
\begin{IEEEeqnarray}{lCl}
  \IEEEeqnarraymulticol{3}{l}{h\Biggl(\left.\sum_{\ell=0}^{k-1}
      H_k^{(\ell)}X_{k-\ell}+Z_k\right|X_1^n=x_1^n,\vect{H}_1^{k-1}\Biggr)}\nonumber\\
    \qquad \qquad \qquad \qquad \quad & \geq & \log\Biggl(\sum_{\ell\in\set{S}_k}
      e^{h\Bigl(H_k^{(\ell)}X_{k-\ell}\Bigm|X_{k-\ell}=x_{k-\ell},\bigl\{H_{k'}^{(\ell)}\bigr\}_{k'=1}^{k-1}\Bigr)}+e^{h(Z_k)}
    \Biggr).\IEEEeqnarraynumspace\label{eq:entropy1}
\end{IEEEeqnarray}
We lower bound the differential entropies on the RHS of
\eqref{eq:entropy1} as follows. The differential entropies in the sum are lower
bounded by
\begin{IEEEeqnarray}{lCl}
  \IEEEeqnarraymulticol{3}{l}{h\biggl(H_k^{(\ell)}X_{k-\ell}\biggm|X_{k-\ell}=x_{k-\ell},\Bigl\{H_{k'}^{(\ell)}\Bigr\}_{k'=1}^{k-1}\biggr)}\nonumber\\
  \qquad\qquad\qquad & = &
  \log\big(\alpha_{\ell}|x_{k-\ell}|^2\big)+h\biggl(H_k^{(\ell)}\biggm|\Big\{H_{k'}^{(\ell)}\Big\}_{k'=1}^{k-1}\biggr)-\log\alpha_{\ell}\IEEEeqnarraynumspace\nonumber\\
  & \geq &
  \log\big(\alpha_{\ell}|x_{k-\ell}|^2\big) +
  \inf_{\ell\in\set{L}}\left(h_{\ell}-\log\alpha_{\ell}\right),\quad\qquad \ell\in\set{S}_k,\label{eq:entropy2}
\end{IEEEeqnarray}
where the equality follows from the behavior of differential entropy under
scaling \cite[Thm.~9.6.4]{coverthomas91}; and where the inequality
follows by the stationarity of the process
$\big\{H_k^{(\ell)},k\in\Naturals\big\}$, which implies that the
differential entropy
\begin{equation*}
h\biggl(H_k^{(\ell)}\biggm|\Big\{H_{k'}^{(\ell)}\Big\}_{k'=1}^{k-1}\biggr),
\quad \ell\in\set{S}_k
\end{equation*}
cannot be smaller than the differential entropy rate $h_{\ell}$
\cite[Thms.~4.2.1 \& 4.2.2]{coverthomas91}, and by lower bounding
$(h_{\ell}-\log\alpha_{\ell})$ by
$\inf_{\ell\in\set{L}}(h_{\ell}-\log\alpha_{\ell})$ (which holds for each
$\ell\in\set{S}_k$ because $\set{S}_k\subseteq\set{L}$).
The last differential entropy on the RHS of \eqref{eq:entropy1} is
lower bounded by
\begin{equation}
  \label{eq:obvious}
  h(Z_k) =\log(\pi e \sigma^2)\geq \inf_{\ell\in\set{L}}\left(h_{\ell}-\log\alpha_{\ell}\right)
  + \log\sigma^2,
\end{equation}
which follows because conditioning cannot increase differential
entropy, and because Gaussian random variables maximize differential
entropy:
\begin{IEEEeqnarray}{lCl}
  \inf_{\ell\in\set{L}}\left(h_{\ell}-\log\alpha_{\ell}\right)
  & \leq & \inf_{\ell\in\set{L}}
  \left(h\Bigl(H_k^{(\ell)}\Bigr)-\log\alpha_{\ell}\right)\nonumber\\
  & \leq & \inf_{\ell\in\set{L}} \bigl(\log(\pi e \alpha_{\ell})-\log\alpha_{\ell}\bigr)\nonumber\\
  & = & \log(\pi e).
\end{IEEEeqnarray}
Applying \eqref{eq:entropy2} \& \eqref{eq:obvious} to
\eqref{eq:entropy1}, and averaging over $X_1^n$, yields then
\begin{IEEEeqnarray}{lCl}
  h\big(Y_k\big|X_1^n,Y_1^{k-1}\big) & \geq & \E{\log\Biggl(\sum_{\ell\in\set{S}_k}
      \alpha_{\ell} |X_{k-\ell}|^2 e^{\inf_{\ell\in\set{L}}\left(h_{\ell}-\log\alpha_{\ell}\right)}+\sigma^2
      e^{\inf_{\ell\in\set{L}}\left(h_{\ell}-\log\alpha_{\ell}\right)}\Biggr)} \nonumber\\
   & = & \E{\log\left(\sigma^2+\sum_{\ell=0}^{k-1} \alpha_{\ell}
      |X_{k-\ell}|^2\right)} + \inf_{\ell\in\set{L}}\left(h_{\ell}-\log\alpha_{\ell}\right).
  \label{eq:U4}
\end{IEEEeqnarray}

Returning to the analysis of \eqref{eq:up1}, we obtain from
\eqref{eq:U3} and \eqref{eq:U4}
\begin{IEEEeqnarray}{lCl}
  I\big(X_1^n;Y_k\big|Y_1^{k-1}\big) & \leq & \frac{1}{2} \E{\log|Y_k|^2} -
  \frac{1}{2}\E{\log|Y_{k-\ell_0}|^2} + \log 2 +
  \E{\log\left(\sigma^2+\sum_{\ell=0}^{k-1}
      \alpha_{\ell}|X_{k-\ell}|^2\right)}\nonumber\\
  & & {} - \E{\log\left(\sigma^2+\sum_{\ell=0}^{k-1} \alpha_{\ell}
      |X_{k-\ell}|^2\right)} -
  \inf_{\ell\in\set{L}}\left(h_{\ell}-\log\alpha_{\ell}\right) + \log
  \frac{\pi^2}{\sqrt{\tilde{\rho}}}\nonumber\\
  & = & \frac{1}{2} \E{\log|Y_k|^2} -
  \frac{1}{2}\E{\log|Y_{k-\ell_0}|^2} + \const{K}, \label{eq:almostfinished}
\end{IEEEeqnarray}
where $\const{K}$ is defined as
\begin{IEEEeqnarray}{lCl}
  \const{K} & \triangleq & \log\frac{2\pi^2}{\sqrt{\tilde{\rho}}} - \inf_{\ell\in\set{L}}\left(h_{\ell}-\log\alpha_{\ell}\right).
\end{IEEEeqnarray}

Applying \eqref{eq:almostfinished} and \eqref{eq:firstl} to
\eqref{eq:chainrule}, we have
\begin{IEEEeqnarray}{lCl}
  \IEEEeqnarraymulticol{3}{l}{\frac{1}{n} I(X_1^n;Y_1^n)}\nonumber\\
  \quad & \leq & \frac{1}{n} \sum_{k=1}^{\ell_0}
  \log\left(1+\sup_{\ell\in\Naturals_0} \alpha_{\ell} \,
    n\,\SNR\right) + \frac{1}{n} \sum_{k=\ell_0+1}^n \biggl( \frac{1}{2} \E{\log|Y_k|^2} -
  \frac{1}{2}\E{\log|Y_{k-\ell_0}|^2} + \const{K}\biggr)\nonumber\\
  & = & \frac{\ell_0}{n} \log\left(1+\sup_{\ell\in\Naturals} \alpha_{\ell} \,
    n\,\SNR\right) + \frac{n-\ell_0}{n} \const{K} + \frac{1}{2n} \sum_{k=\ell_0+1}^n \biggl(\E{\log|Y_k|^2} -
  \E{\log|Y_{k-\ell_0}|^2}\biggr).\IEEEeqnarraynumspace \label{eq:almostfinished2}
\end{IEEEeqnarray}

To show that the RHS of \eqref{eq:almostfinished2} is bounded
in the SNR, we use that, for any sequences $\{a_k\}$ and $\{b_k\}$,
\begin{IEEEeqnarray}{lCl}
  \sum_{k=\ell_0+1}^n (a_k-b_k) & = & \sum_{k=n-\ell_0+1}^n
 (a_k-b_{k-n+2\ell_0}) +
 \sum_{k=\ell_0+1}^{n-\ell_0}(a_k-b_{k+\ell_0}).\IEEEeqnarraynumspace
 \label{eq:sum2}
\end{IEEEeqnarray}
Defining
\begin{equation}
  a_k \triangleq \E{\log|Y_k|^2}
\end{equation}
and
\begin{equation}
  b_k \triangleq \E{\log|Y_{k-\ell_0}|^2}
\end{equation}
we have for the first sum on the RHS of \eqref{eq:sum2}
\begin{IEEEeqnarray}{lCl}
  \sum_{k=n-\ell_0+1}^n
  (a_k-b_{k-n+2\ell_0}) & = & \sum_{k=n-\ell_0+1}^n \biggl(\E{\log|Y_k|^2} -
  \E{\log|Y_{k-n+\ell_0}|^2}\biggr) \nonumber\\
  & \leq & \sum_{k=n-\ell_0+1}^n \biggl(\log\E{|Y_k|^2} -
  \E{\log|Y_{k-n+\ell_0}|^2}\biggr)\nonumber\\
  & \leq & \sum_{k=n-\ell_0+1}^n \Biggl(
  \log\biggl(\sigma^2+\sup_{\ell\in\Naturals_0} \alpha_{\ell} \, n\,\const{P}\biggr) -
  \E{\log|Y_{k-n+\ell_0}|^2}\Biggr)\nonumber\\
  & \leq & \sum_{k=n-\ell_0+1}^n \Biggl(
  \log\biggl(\sigma^2+\sup_{\ell\in\Naturals_0} \alpha_{\ell} \, n\,\const{P}\biggr)-\E{\log|Z_{k-n+\ell_0}|^2}\Biggr)\nonumber\\
  & = & \ell_0 \log\biggl(1+\sup_{\ell\in\Naturals_0}\alpha_{\ell} \, n \, \SNR\biggr)+\ell_0\gamma,\label{eq:last1}
\end{IEEEeqnarray}
where $\gamma\approx 0.577$ denotes Euler's constant. Here the first
inequality follows by Jensen's inequality; the following inequality
follows by upper bounding
\begin{equation*}
  \E{|Y_k|^2} = \sigma^2+\sum_{\ell=0}^{k-1} \alpha_{\ell}
  \E{|X_{k-\ell}|^2} \leq \sigma^2 + \sup_{\ell\in\Naturals_0}
  \alpha_{\ell} \, n \,\const{P};
\end{equation*}
the subsequent inequality follows by noting that, conditional on
$\sum_{\ell=0}^{k-n+\ell_0-1} H_{k-n+\ell_0}^{(\ell)} X_{k-n+\ell_0-\ell}$,
we have that $|Y_{k-n+\ell_0}|^2$ is stochastically larger than
$|Z_{k-n+\ell_0}|^2$, so
\begin{IEEEeqnarray*}{l}
  \Econdd{\log|Y_{k-n+\ell_0}|^2}{\sum_{\ell=0}^{k-n+\ell_0-1}H_{k-n+\ell_0}^{(\ell)} X_{k-n+\ell_0-\ell}}\\
  \qquad\qquad\qquad\qquad\geq \Econdd{\log|Z_{k-n+\ell_0}|^2}{\sum_{\ell=0}^{k-n+\ell_0-1}H_{k-n+\ell_0}^{(\ell)} X_{k-n+\ell_0-\ell}}
\end{IEEEeqnarray*}
from which we obtain the lower bound
$\E{\log|Y_{k-n+\ell_0}|^2}\geq\E{\log|Z_{k-n+\ell_0}|^2}$ upon
averaging over $\sum_{\ell=0}^{k-n+\ell_0-1} H_{k-n+\ell_0}^{(\ell)} X_{k-n+\ell_0-\ell}$
 (see \cite[Sec.~VI--B]{lapidothmoser03_3} and in particular
 \cite[Lemma 6.2 b)]{lapidothmoser03_3}); and the last equality
 follows by evaluating the expected logarithm of an exponentially
 distributed random variable of mean $\sigma^2$, i.e.,
 $\E{\log|Z_{k-n+\ell_0}|^2}=\log\sigma^2-\gamma$. 

For the second sum on the RHS of \eqref{eq:sum2} we have
\begin{IEEEeqnarray}{lCl}
  \sum_{k=\ell_0+1}^{n-\ell_0}(a_k-b_{k+\ell_0}) & = &
  \sum_{k=\ell_0+1}^{n-\ell_0} \biggl(\E{\log|Y_k|^2} -
  \E{\log|Y_{k}|^2}\biggr)=0.\label{eq:last2}
\end{IEEEeqnarray}
Thus applying \eqref{eq:sum2}--\eqref{eq:last2} to
\eqref{eq:almostfinished2} yields
\begin{IEEEeqnarray}{lCl}
  \frac{1}{n}I(X_1^n;Y_1^n) & \leq & \frac{2 \ell_0}{n} \log\left(1+\sup_{\ell\in\Naturals_0} \alpha_{\ell} \,
    n\,\SNR\right) + \frac{n-\ell_0}{n} \const{K} + \frac{\ell_0}{n} \gamma,
\end{IEEEeqnarray}
which tends to
\begin{equation*}
  \const{K} = \log\frac{2\pi^2}{\sqrt{\tilde{\rho}}} -
  \inf_{\ell\in\set{L}}\left(h_{\ell}-\log\alpha_{\ell}\right)
\end{equation*}
as $n$ tends to infinity. This proves Part~(i) of Proposition~\ref{prop:upperbounds}.

\subsection{Unbounded Capacity}
\label{sub:upper2}
We prove Part~(ii) of
Proposition~\ref{prop:upperbounds} by deriving an upper bound on
capacity that holds under the assumption \eqref{eq:alphasumlower}, namely,
\begin{equation*}
  \sum_{\ell=0}^{\infty} \alpha_{\ell} < \infty.
\end{equation*}
From this upper bound follows that
\begin{equation}
  \varlimsup_{\SNR\to\infty} \bigl\{C(\SNR) -
  \log\log\SNR\bigr\} < \infty,
\end{equation}
which in turn shows that the capacity pre-loglog is upper bounded by
\begin{equation}
  \Lambda \triangleq \varlimsup_{\SNR\to\infty}
  \frac{C(\SNR)}{\log\log\SNR} \leq 1.
\end{equation}
This yields the converse to Theorem~\ref{thm:finite}.

As in Section~\ref{sub:upper1}, the desired upper bound follows by
\eqref{eq:capacity} and by deriving an upper bound on $\frac{1}{n} I(X_1^n;Y_1^n)$. To this end,
we begin with the chain rule for mutual information
\begin{equation}
  I\big(X_1^n;Y_1^n\big) = \sum_{k=1}^n I\big(X_1^n;Y_k\big|Y_1^{k-1}\big)\label{eq:upperchain}
\end{equation}
and upper bound each summand on the RHS of
\eqref{eq:upperchain} using \cite[Eq.~(27)]{lapidothmoser03_3}
\begin{IEEEeqnarray}{lCl}
  I\big(X_1^n;Y_k\big|Y_1^{k-1}\big)& \leq & \E{\log
    |Y_k|^2}-h\big(Y_k\big|X_1^n,Y_1^{k-1}\big) + \xi
  \bigl(1+\log\E{|Y_k|^2}-\E{\log|Y_k|^2}\bigr)\nonumber\\
  & & {} +\log\Gamma(\xi)-\xi\log\xi+\log\pi \nonumber\\
  & = & (1-\xi)\E{\log
    |Y_k|^2}  -h\big(Y_k\big|X_1^n,Y_1^{k-1}\big) + \xi
  \bigl(1+\log\E{|Y_k|^2}\bigr)\nonumber\\
  & & {} +\log\Gamma(\xi)-\xi\log\xi+\log\pi, \label{eq:duality2}
\end{IEEEeqnarray}
for any $\xi>0$. Here $\Gamma(\cdot)$ denotes the Gamma
function.

We evaluate the terms on the RHS of \eqref{eq:duality2} individually.
We upper bound the first term using Jensen's inequality
\begin{IEEEeqnarray}{lCl}
  \E{\log|Y_k|^2} & = & \E{\Econd{\log|Y_k|^2}{X_1^k}}\nonumber\\
  & \leq & \E{\log \Econd{|Y_k|^2}{X_1^k}} \nonumber\\
  & = &\E{\log\left(\sigma^2+\sum_{\ell=0}^{k-1}\alpha_{\ell}|X_{k-\ell}|^2\right)}.\label{eq:upper1}
\end{IEEEeqnarray}
The second term was already evaluated in \eqref{eq:U4}
\begin{IEEEeqnarray}{lCl}
  h\big(Y_k\big|X_1^n,Y_1^{k-1}\big) & \geq &
  \E{\log\left(\sigma^2+\sum_{\ell=0}^{k-1}\alpha_{\ell}|X_{k-\ell}|^2\right)}
  + \inf_{\ell\in\set{L}}\left(h_{\ell}-\alpha_{\ell}\right),\label{eq:upper2}
\end{IEEEeqnarray}
and the next term is readily evaluated as
\begin{equation}
  \log \E{|Y_k|^2} = \log\left(\sigma^2+\sum_{\ell=0}^{k-1}\alpha_{\ell}\E{|X_{k-\ell}|^2}\right).\label{eq:upper3}
\end{equation}

Our choice of $\xi$ will satisfy $\xi<1$ (see \eqref{eq:xi} ahead). We
therefore obtain, upon substituting \eqref{eq:upper1}--\eqref{eq:upper3} in
\eqref{eq:duality2},
\begin{IEEEeqnarray}{lCl}
  I\big(X_1^n;Y_k\big|Y_1^{k-1}\big) & \leq & (1-\xi)
  \E{\log\left(\sigma^2+\sum_{\ell=0}^{k-1}\alpha_{\ell}|X_{k-\ell}|^2\right)}
  -  \E{\log\left(\sigma^2+\sum_{\ell=0}^{k-1}\alpha_{\ell}|X_{k-\ell}|^2\right)}\nonumber\\
  & & {} -\inf_{\ell\in\set{L}}\left(h_{\ell}-\alpha_{\ell}\right) + \xi\Biggl(
  1+\log\Biggl(\sigma^2+\sum_{\ell=0}^{k-1}\alpha_{\ell}\E{|X_{k-\ell}|^2}\Biggr)\Biggr) \nonumber\\
  & & {} + \log\Gamma(\xi)-\xi\log\xi+\log\pi\nonumber\\
  & = & -\inf_{\ell\in\set{L}}\left(h_{\ell}-\alpha_{\ell}\right)
  \nonumber\\
  & & {} + \xi\left(
    1+\log\Biggl(\sigma^2+\sum_{\ell=0}^{k-1}\alpha_{\ell}\E{|X_{k-\ell}|^2}\Biggr)-\E{\log\left(\sigma^2+\sum_{\ell=0}^{k-1}\alpha_{\ell}|X_{k-\ell}|^2\right)}\right)\nonumber\\
  & & {} + \log\Gamma(\xi)-\xi\log\xi+\log\pi\nonumber\\
  & \leq & -\inf_{\ell\in\set{L}}\left(h_{\ell}-\alpha_{\ell}\right) +
  \log\Gamma(\xi)-\xi\log\xi+\log\pi\nonumber\\
  & & {} + \xi\left(
    1+\log\Biggl(1+\sum_{\ell=0}^{k-1}\alpha_{\ell}\E{|X_{k-\ell}|^2}/\sigma^2\Biggr)\right),\label{eq:upper6}
\end{IEEEeqnarray}
where the last inequality follows by lower bounding
$\E{\log\left(\sigma^2+\sum_{\ell=0}^{k-1}\alpha_{\ell}|X_{k-\ell}|^2\right)}\geq\log\sigma^2$.

We choose
\begin{equation}
  \label{eq:xi}
  \xi = \frac{1}{1+\log\bigl(1+\alpha \,\SNR\bigr)}
\end{equation}
(where $\alpha$ was defined in \eqref{eq:alphasumlower}).
Defining
\begin{equation}
  \Psi(\SNR) \triangleq \left[\log\Gamma(\xi)-\log\frac{1}{\xi}-\xi\log\xi\right|_{\xi=\bigl(1+\log(1+\alpha \,\SNR)\bigr)^{-1}},
\end{equation}
we obtain
\begin{IEEEeqnarray}{lCl}
  I\big(X_1^n;Y_k\big|Y_1^{k-1}\big) & \leq & -\inf_{\ell\in\set{L}}
  \left(h_{\ell}-\alpha_{\ell}\right) +
  \log\bigl(1+\log(1+\alpha\,\SNR)\bigr)+\Psi(\SNR) + \log\pi \nonumber\\
  & & {} +
  \frac{1+\log\left(1+\sum_{\ell=0}^{k-1}\alpha_{\ell}\E{|X_{k-\ell}|^2}/\sigma^2\right)}{1+\log\left(1+\alpha\,
  \SNR\right)}.\label{eq:uppersummand}
\end{IEEEeqnarray}
Using \eqref{eq:uppersummand} in \eqref{eq:upperchain} yields then
\begin{IEEEeqnarray}{lCl}
  \frac{1}{n} I\big(X_1^n;Y_1^n\big) & \leq & -\inf_{\ell\in\set{L}}
  \left(h_{\ell}-\alpha_{\ell}\right) +
  \log\bigl(1+\log(1+\alpha\,\SNR)\bigr)+\Psi(\SNR) + \log\pi
  \nonumber\\
  & & {} +  \frac{1+\frac{1}{n}\sum_{k=1}^n \log\left(1+\sum_{\ell=0}^{k-1}\alpha_{\ell}\E{|X_{k-\ell}|^2}/\sigma^2\right)}{1+\log\left(1+\alpha\,
  \SNR\right)}.\label{eq:UPPER}
\end{IEEEeqnarray}
By Jensen's inequality we have
\begin{IEEEeqnarray}{lCl}
  \frac{1}{n}\sum_{k=1}^n
    \log\left(1+\sum_{\ell=0}^{k-1}\alpha_{\ell}\E{|X_{k-\ell}|^2}/\sigma^2\right)
    & \leq &
    \log\left(1+\frac{1}{n}\sum_{k=1}^n\sum_{\ell=0}^{k-1}\alpha_{\ell}\E{|X_{k-\ell}|^2}/\sigma^2\right)\nonumber\\
  & \leq & \log\left(1+\alpha\,\SNR\right),\label{eq:upperjensen}
\end{IEEEeqnarray}
where the last inequality follows by rewriting the double sum as
\begin{equation*}
  \frac{1}{n}\sum_{k=1}^n\E{|X_{k}|^2}/\sigma^2
  \sum_{\ell=0}^{n-k}\alpha_{\ell},
\end{equation*}
and by upper bounding then $\sum_{\ell=0}^{k-n}\alpha_{\ell}\leq
\alpha$ and using the power constraint \eqref{eq:power}.

Combining \eqref{eq:upperjensen} and \eqref{eq:UPPER} with
\eqref{eq:capacity}, we obtain the upper bound
\begin{IEEEeqnarray}{lCl}
  C(\SNR) & \leq & -\inf_{\ell\in\set{L}}
  \left(h_{\ell}-\alpha_{\ell}\right) +
  \log\bigl(1+\log(1+\alpha\,\SNR)\bigr)+\Psi(\SNR) + \log\pi +
  1.\label{eq:LastFinite}
\end{IEEEeqnarray}
It follows by \cite[Eq.~(337)]{lapidothmoser03_3} that
\begin{equation}
  \label{eq:finitetozero}
  \lim_{\SNR\to\infty} \Psi(\SNR) = \lim_{\xi\downarrow 0}
  \Bigl\{\log\Gamma(\xi)-\log\frac{1}{\xi}-\xi\log\xi\Bigr\} = 0.
\end{equation}
Noting that
\begin{equation*}
  \lim_{\SNR\to\infty}
  \Bigl\{\log\bigl(1+\log(1+\alpha\,\SNR)\bigr)-\log\log\SNR\Bigr\}
  = 0,
\end{equation*}
we obtain from \eqref{eq:LastFinite} and \eqref{eq:finitetozero} the desired result
\begin{equation}
  \varlimsup_{\SNR\to\infty}
  \bigl\{C(\SNR)-\log\log\SNR\bigr\} \leq 1 +
  \log\pi -\inf_{\ell\in\set{L}}
  \left(h_{\ell}-\alpha_{\ell}\right).
\end{equation}

\section{Proofs of the Achievability Results}
\label{sec:lowerproof}

In Section~\ref{sub:lower}, we derive the lower bound on channel
capacity that is presented in Proposition~\ref{prop:lowerbound}. This
lower bound will be used in Sections~\ref{sub:neccondbounded} \& \ref{sub:preloglog} to prove
Part~(ii) of Theorem~\ref{thm:bounded} and to prove the direct part of
Theorem~\ref{thm:finite}, respectively.

\subsection{Lower Bound}
\label{sub:lower}

To derive the desired lower bound on capacity, we evaluate
$\frac{1}{n}I(X_1^n;Y_1^n)$ for the following distribution on the
inputs $\{X_k\}$.

Let $\const{L}(\const{P})$ be such that
\begin{equation}
  \label{eq:L}
  \sum_{\ell=\const{L}(\const{P})+1}^\infty \alpha_{\ell}\,\const{P} \leq \sigma^2.
\end{equation}
To shorten notation, we shall write in the following $\const{L}$ instead
of $\const{L}(\const{P})$. Let $\tau\in\Naturals$ be some positive
integer that possibly depends on $\const{L}$, and let
$\vect{X}_b=(X_{b(\const{L}+\tau)+1},\ldots,X_{(b+1)(\const{L}+\tau)})$.
We choose $\{\vect{X}_b\}$ to be IID with
\begin{equation*}
  \vect{X}_b = \bigl(\underbrace{0,\ldots,0}_{\const{L}},\tilde{X}_{b\tau+1},\ldots,\tilde{X}_{(b+1)\tau}\bigr),
\end{equation*}
where $\tilde{X}_{b\tau+1},\ldots,\tilde{X}_{(b+1)\tau}$ is a sequence
of independent, zero-mean, circularly-symmetric, complex random variables with
$\log|\tilde{X}_{b\tau+\nu}|^2$ being uniformly distributed over the
interval $\bigl[\log \const{P}^{(\nu-1)/\tau},\log
\const{P}^{\nu/\tau}\bigr]$, i.e., for each $\nu=1,\ldots,\tau$
\begin{equation*}
  \log |\tilde{X}_{b\tau+\nu}|^2 \sim
  \Uniform{\bigl[\log\const{P}^{(\nu-1)/\tau},\log
  \const{P}^{\nu/\tau}\bigr]}.
\end{equation*}
(Here and throughout this proof we assume that $\const{P}>1$.)

Let $\kappa\triangleq \lfloor \frac{n}{\const{L}+\tau}\rfloor$ (where
$\lfloor a \rfloor$ denotes the largest integer that is less
than or equal to $a$), and
let $\vect{Y}_b$ denote the vector
$(Y_{b(L+\tau)+1},\ldots,Y_{(b+1)(\const{L}+\tau)})$. By the chain
rule for mutual information we have
\begin{IEEEeqnarray}{lCl}
  I\bigl(X_1^n;Y_1^n\bigr) & \geq &
  I\bigl(\vect{X}_0^{\kappa-1};\vect{Y}_0^{\kappa-1}\bigr)\nonumber\\
  & = & \sum_{b=0}^{\kappa-1}
  I\bigl(\vect{X}_b;\vect{Y}_0^{\kappa-1}\bigm|\vect{X}_0^{b-1}\bigr)\nonumber\\
  & \geq & \sum_{b=0}^{\kappa-1} I(\vect{X}_b;\vect{Y}_b),\label{eq:proof1}
\end{IEEEeqnarray}
where the first inequality follows by restricting the number of
observables; and where the last inequality follows by restricting the
number of observables and by noting that $\{\vect{X}_b\}$ is IID.

We continue by lower bounding each summand on the RHS of
\eqref{eq:proof1}. We use again the chain rule and that reducing
observations cannot increase mutual information to obtain
\begin{IEEEeqnarray}{lCl}
  I(\vect{X}_b;\vect{Y}_b) & = & \sum_{\nu=1}^{\tau}
  I\bigl(\tilde{X}_{b\tau+\nu};\vect{Y}_b\bigm|\tilde{X}_{b\tau+1}^{b\tau+\nu-1}\bigr)\nonumber\\
  & \geq & \sum_{\nu=1}^{\tau} I\bigl(\tilde{X}_{b\tau+\nu};Y_{b(\const{L}+\tau)+\const{L}+\nu}\bigm|\tilde{X}_{b\tau+1}^{b\tau+\nu-1}\bigr)\nonumber\\
  & \geq & \sum_{\nu=1}^{\tau} I\bigl(\tilde{X}_{b\tau+\nu};Y_{b(\const{L}+\tau)+\const{L}+\nu}\bigr),\label{eq:proof1b}
\end{IEEEeqnarray}
where we have additionally used in the last inequality that
$\tilde{X}_{b\tau+1},\ldots,\tilde{X}_{(b+1)\tau}$ are independent.

Defining
\begin{IEEEeqnarray}{lCl}
  W_{b\tau+\nu} & \triangleq &
  \sum_{\ell=1}^{b(\const{L}+\tau)+\const{L}+\nu-1}
  H_{b(\const{L}+\tau)+\const{L}+\nu}^{(\ell)}
  X_{b(\const{L}+\tau)+\const{L}+\nu-\ell} + Z_{b(\const{L}+\tau)+\const{L}+\nu},
\end{IEEEeqnarray}
each summand on the RHS of \eqref{eq:proof1b} can be written as
\begin{IEEEeqnarray}{lCl}
  I\bigl(\tilde{X}_{b\tau+\nu};Y_{b(\const{L}+\tau)+\const{L}+\nu}\bigr)
  & = & I\bigl(\tilde{X}_{b\tau+\nu};H_{b(\const{L}+\tau)+\const{L}+\nu}^{(0)}\tilde{X}_{b\tau+\nu}+W_{b\tau+\nu}\bigr).\IEEEeqnarraynumspace\label{eq:proof2}
\end{IEEEeqnarray}
A lower bound on \eqref{eq:proof2} follows from the following lemma.
\begin{lemma}
  \label{lemma}
  Let the random variables $X$, $H$, and $W$ have finite second
  moments. Assume that both $X$ and $H$ are of finite differential
  entropy. Finally, assume that $X$ is independent of $H$; that
  $X$ is independent of $W$; and that $X \markov H \markov W$ forms a
  Markov chain. Then
  \begin{IEEEeqnarray}{lCl}
    I(X;HX+W) & \geq & h(X)-\E{\log|X|^2}+\E{\log|H|^2} - \E{\log\biggl(\pi e\biggl(\sigma_H+\frac{\sigma_W}{|X|}\biggr)^2\biggr)},\IEEEeqnarraynumspace
  \end{IEEEeqnarray}
  where $\sigma^2_H\geq 0$ and $\sigma^2_H>0$ denote the variances of\/
  $W$ and $H$. (Note that the assumptions that $X$ and $H$ have finite
  second moments and are of finite differential entropy guarantee
  that $\E{\log|X|^2}$ and $\E{\log|H|^2}$ are finite, see \textnormal{\cite[Lemma
  6.7e]{lapidothmoser03_3}}.)
\end{lemma}
\begin{proof}
  See \cite[Lemma 4]{lapidoth05_2}.
\end{proof}
It can be easily verified that for the channel model given in
Section~\ref{sec:channelmodel} and for the above coding scheme the lemma's
conditions are satisfied. We therefore obtain from Lemma~\ref{lemma}
\begin{IEEEeqnarray}{lCl}
  I\bigl(\tilde{X}_{b\tau+\nu};H_{b(\const{L}+\tau)+\const{L}+\nu}^{(0)}\tilde{X}_{b\tau+\nu}+W_{b\tau+\nu}\bigr)
  & \geq & h\bigl(\tilde{X}_{b\tau+\nu}\bigr)
  -\E{\log|\tilde{X}_{b\tau+\nu}|^2} +
  \E{\log\bigl|H_{b(\const{L}+\tau)+\const{L}+\nu}^{(0)}\bigr|^2}\nonumber\\
  & & {} - \E{\log\biggl(\pi e\biggl(\sqrt{\alpha_0}+\frac{\sqrt{\E{|W_{b\tau+\nu}|^2}}}{|\tilde{X}_{b\tau+\nu}|}\biggr)^2\biggr)}.\label{eq:proof3}
\end{IEEEeqnarray}
Using that the differential entropy of a circularly-symmetric random
variable is given by (see \cite[Eqs.~(320) \&
(316)]{lapidothmoser03_3})
\begin{equation}
  h\bigl(\tilde{X}_{b\tau+\nu}\bigr) =
  \E{\log|\tilde{X}_{b\tau+\nu}|^2}+h\bigl(\log|\tilde{X}_{b\tau+\nu}|^2\bigr)
  + \log\pi,
\end{equation}
and evaluating $h(\log|\tilde{X}_{b\tau+\nu}|^2)$ for our choice of
$\tilde{X}_{b\tau+\nu}$, yields for the first two terms on the RHS of
\eqref{eq:proof3}
\begin{equation}
  \label{eq:proof3b}
  h\bigl(\tilde{X}_{b\tau+\nu}\bigr)
  -\E{\log|\tilde{X}_{b\tau+\nu}|^2} = \log\log\const{P}^{1/\tau} +\log\pi.
\end{equation}
We next upper bound
\begin{IEEEeqnarray}{lCl}
  \frac{\E{|W_{b\tau+\nu}|^2}}{|\tilde{X}_{b\tau+\nu}|^2} & = & 
  \sum_{\ell=1}^{\const{L}} \alpha_{\ell}
  \frac{\E{|X_{b(\const{L}+\tau)+\const{L}+\nu-\ell}|^2}}{|\tilde{X}_{b\tau+\nu}|^2}
  + \sum_{\ell=\const{L}+1}^{b(\const{L}+\tau)+\const{L}+\nu-1} \alpha_{\ell}
  \frac{\E{|X_{b(\const{L}+\tau)+\const{L}+\nu-\ell}|^2}}{|\tilde{X}_{b\tau+\nu}|^2}\nonumber\\
  & & {} +\frac{\sigma^2}{|\tilde{X}_{b\tau+\nu}|^2}.\IEEEeqnarraynumspace\label{eq:proof4}
\end{IEEEeqnarray}
To this end, we note that for our choice of $\{X_k\}$ and by the
assumption that $\const{P}>1$, we have
\begin{equation}
  \E{|X_{\ell}|^2} \leq
  \const{P},\quad \ell\in\Naturals,
\end{equation}
\begin{equation}
\E{|X_{b(\const{L}+\tau)+\const{L}+\nu-\ell}|^2}\leq
\const{P}^{(\nu-\ell)/\tau}, \quad \ell=1,\ldots,\const{L},
\end{equation}
 and
\begin{equation}
  \label{eq:proof4a}
  |\tilde{X}_{b\tau+\nu}|^2\geq\const{P}^{(\nu-1)/\tau}\geq 1,
\end{equation}
from which we obtain
\begin{equation}
  \label{eq:proof4b}
  \frac{\E{|X_{b(\const{L}+\tau)+\const{L}+\nu-\ell}|^2}}{|\tilde{X}_{b\tau+\nu}|^2}
  \leq
  \frac{\const{P}^{(\nu-\ell)/\tau}}{\const{P}^{(\nu-1)/\tau}}\leq 1,
  \quad \ell=1,\ldots,\const{L}
\end{equation}
and
\begin{equation}
  \label{eq:proof4c}
  \frac{\E{|X_{b(\const{L}+\tau)+\const{L}+\nu-\ell}|^2}}{|\tilde{X}_{b\tau+\nu}|^2}
  \leq \const{P}, \quad \ell=\const{L}+1,\ldots,b(\const{L}+\tau)+\const{L}+\nu-1.
\end{equation}
Applying \eqref{eq:proof4a}--\eqref{eq:proof4c} to \eqref{eq:proof4}
yields
\begin{IEEEeqnarray}{lCl}
  \frac{\E{|W_{b\tau+\nu}|^2}}{|\tilde{X}_{b\tau+\nu}|^2}
  & \leq & \sum_{\ell=1}^{\const{L}}\alpha_{\ell} +
  \sum_{\ell=\const{L}+1}^{b(\const{L}+\tau)+\const{L}+\nu-1}
  \alpha_{\ell} \, \const{P}+\sigma^2\nonumber\\
  & \leq & \alpha +
  \sum_{\ell=\const{L}+1}^{\infty}
  \alpha_{\ell} \, \const{P} + \sigma^2 \nonumber\\
  & \leq & \alpha + 2\sigma^2,\label{eq:proof5}
\end{IEEEeqnarray}
with $\alpha$ being defined in \eqref{eq:alphasumlower}. Here
the second inequality follows because $\alpha_{\ell}$,
$\ell\in\Naturals_0$ and $\const{P}$
are nonnegative, and the last inequality follows from
\eqref{eq:L}.

By combining \eqref{eq:proof3} with \eqref{eq:proof3b} \&
\eqref{eq:proof5}, and by noting that by the stationarity of
$\bigl\{H_k^{(0)},\;k\in\Naturals\bigr\}$
\begin{equation*}
  \E{\log\bigl|H_{b(\const{L}+\tau)+\const{L}+\nu}^{(0)}\bigr|^2} =
  \E{\log\bigl|H_{1}^{(0)}\bigr|^2},
\end{equation*}
we obtain the lower bound
\begin{IEEEeqnarray}{lCl}
  I\bigl(\tilde{X}_{b\tau+\nu};H_{b(\const{L}+\tau)+\const{L}+\nu}^{(0)}\tilde{X}_{b\tau+\nu}+W_{b\tau+\nu}\bigr)
  & \geq & \log\log\const{P}^{1/\tau} +
  \E{\log\bigl|H_{1}^{(0)}\bigr|^2} - 1 \nonumber\\
  & & {} - 2\log\bigl(\sqrt{\alpha_0}+\sqrt{\alpha + 2\sigma^2}\bigr).\IEEEeqnarraynumspace\label{eq:proof6}
\end{IEEEeqnarray}
Note that the RHS of \eqref{eq:proof6} neither depends on $\nu$ nor on
$b$. We therefore have from \eqref{eq:proof6}, \eqref{eq:proof1b}, and
\eqref{eq:proof1}
\begin{equation}
  I\bigl(X_1^n;Y_1^n\bigr) \geq \kappa\tau \log\log\const{P}^{1/\tau}
  + \kappa\tau \Upsilon,\label{eq:proof7}
\end{equation}
where we define $\Upsilon$ as
\begin{equation}
  \label{eq:upsilon}
  \Upsilon \triangleq \E{\log\bigl|H_{1}^{(0)}\bigr|^2} - 1 - 2\log\bigl(\sqrt{\alpha_0}+\sqrt{\alpha + 2\sigma^2}\bigr).
\end{equation}
Dividing the RHS of \eqref{eq:proof7} by $n$, and computing the limit
as $n$ tends to infinity, yields the lower bound
\begin{equation}
  \label{eq:lowerbound}
  C(\SNR) \geq \frac{\tau}{\const{L}+\tau}
  \log\log\const{P}^{1/\tau}+\frac{\tau}{\const{L}+\tau}\Upsilon,
  \quad \const{P}>1,
\end{equation}
where we have used that
$\lim_{n\to\infty}\kappa/n=1/(\const{L}+\tau)$. This proves
Proposition~\ref{prop:lowerbound}.

\subsection{Condition for Unbounded Capacity}
\label{sub:neccondbounded}
We use Proposition~\ref{prop:lowerbound} to prove Part~(ii) of
 Theorem~\ref{thm:bounded}. In particular, we show that if 
\begin{equation}
  \label{eq:conditionThm1ii}
  \lim_{\ell\to\infty} \frac{1}{\ell}\log\frac{1}{\alpha_{\ell}} = \infty,
\end{equation}
then, by cleverly choosing $\const{L}(\const{P})$ and $\tau$, the
lower bound \eqref{eq:proplowerbound}, namely,
\begin{equation*}
  C(\SNR) \geq \frac{\tau}{\const{L}(\const{P})+\tau}
  \log\log\const{P}^{1/\tau}+\frac{\tau}{\const{L}(\const{P})+\tau}\Upsilon,
  \quad \const{P}>1
\end{equation*}
(where $\Upsilon$ is defined in \eqref{eq:upsilon}), can be made
arbitrarily large as $\SNR$ tends to
infinity. To this end, we first note that \eqref{eq:conditionThm1ii}
implies that for every $0<\varrho<1$ we can find an
$\ell_0\in\Naturals$ such that
\begin{equation}
  \label{eq:conditionThmii2}
  \alpha_{\ell}<\varrho^{\ell}, \quad \ell\geq\ell_0.
\end{equation}
By choosing
\begin{equation}
  \label{eq:Lchoice}
  \const{L}(\const{P}) =
  \left\lceil\frac{\log\bigl(\const{P}/\sigma^2\,\varrho/(1-\varrho)\bigr)}{\log(1/\varrho)}\right\rceil
\end{equation}
(where $\lceil a \rceil$ denotes the smallest integer that is greater
than or equal to $a$) and $\tau=\const{L}(\const{P})$, we obtain from
\eqref{eq:proplowerbound} the lower bound
\begin{IEEEeqnarray}{lCl}
  C(\SNR) & \geq & \frac{1}{2}
  \log\frac{\log\const{P}}{\biggl\lceil\frac{\log\bigl(\const{P}/\sigma^2\,\varrho/(1-\varrho)\bigr)}{\log(1/\varrho)}\biggr\rceil}
  + \frac{1}{2}\Upsilon, \quad \const{P}>1.
\end{IEEEeqnarray}
Taking the limit as $\SNR$ (and hence also $\const{P}=\sigma^2\SNR$) tends to
infinity, yields
\begin{equation}
  \lim_{\SNR\to\infty} C(\SNR) \geq \frac{1}{2}
  \log\log\frac{1}{\varrho} + \frac{1}{2} \Upsilon.
\end{equation}
Since this holds for every $0<\varrho<1$
\begin{equation}
  \sup_{\SNR>0} C(\SNR) = \infty.
\end{equation}

It remains to show that $\{\alpha_{\ell}\}$ and our choice of
$\const{L}(\const{P})$ \eqref{eq:Lchoice} satisfy the conditions
\eqref{eq:alphasumlower} \& \eqref{eq:Lcondition} of
Proposition~\ref{prop:lowerbound}, namely,
\begin{equation*}
  \sum_{\ell=0}^{\infty} \alpha_{\ell} < \infty \qquad
  \textnormal{and} \qquad \sum_{\ell=\const{L}(\const{P})+1}^{\infty}
  \alpha_{\ell}\,\const{P} \leq \sigma^2.
\end{equation*}
It follows immediately from \eqref{eq:alpha} and \eqref{eq:conditionThmii2} that
$\{\alpha_{\ell}\}$ satisfies the first condition \eqref{eq:alphasumlower}:
\begin{equation}
  \sum_{\ell=0}^{\infty} \alpha_{\ell} = \sum_{\ell=0}^{\ell_0-1}
  \alpha_{\ell}+\sum_{\ell=\ell_0}^{\infty}\alpha_{\ell} < \ell_0\,\sup_{\ell\in\Naturals_0}\alpha_{\ell}+\sum_{\ell=\ell_0}^{\infty}
  \varrho^{\ell} = \ell_0\,\sup_{\ell\in\Naturals_0}\alpha_{\ell} + \frac{\varrho^{\ell_0}}{1-\varrho} < \infty.
\end{equation}
In order to show that $\const{L}(\const{P})$ satisfies the second
condition \eqref{eq:Lcondition}, we first note that by
\eqref{eq:conditionThmii2}
\begin{equation}
  \sum_{\ell=\ell'+1}^{\infty} \alpha_{\ell} <
  \sum_{\ell=\ell'+1}^{\infty} \varrho^{\ell} =
  \varrho^{\ell'}\frac{\varrho}{1-\varrho}, \quad \ell'\geq\ell_0-1.\label{eq:Thm1ii1}
\end{equation}
Since $\const{L}(\const{P})$ tends to infinity as
$\const{P}\to\infty$, it follows that $\const{L}(\const{P})$ is
greater than $(\ell_0-1)$  for sufficiently large
$\const{P}$. Furthermore, \eqref{eq:Lchoice} implies
\begin{equation}
  \label{eq:smallersigma}
  \varrho^{\const{L}(\const{P})}\frac{\varrho}{1-\varrho}\,\const{P}\leq
  \sigma^2.
\end{equation}
We therefore obtain from \eqref{eq:Thm1ii1} and \eqref{eq:smallersigma}
\begin{equation}
  \sum_{\ell=\const{L}(\const{P})+1}^{\infty} \alpha_{\ell}\,\const{P} <
  \varrho^{\const{L}(\const{P})}\frac{\varrho}{1-\varrho}\,\const{P}\leq
  \sigma^2,
\end{equation}
thus demonstrating that $\const{L}(\const{P})$ satisfies
\eqref{eq:Lcondition}.

\subsection{The Pre-LogLog}
\label{sub:preloglog}
We use Proposition~\ref{prop:lowerbound} to prove
Theorem~\ref{thm:finite}. To this end, we first note that because the
number of paths is finite, we have for some $\const{L}\in\Naturals_0$
\begin{equation}
  \alpha_{\ell}=0, \quad \ell>\const{L},
\end{equation}
which implies that
\begin{equation}
  \sum_{\ell=0}^{\infty} \alpha_{\ell} = \sum_{\ell=0}^{\const{L}}
  \alpha_{\ell} \leq (\const{L}+1) \, \sup_{\ell\in\Naturals_0} \alpha_{\ell}
  < \infty
\end{equation}
and
\begin{equation}
  \sum_{\ell=\const{L}+1}^{\infty} \alpha_{\ell}\const{P} = 0 \leq
  \sigma^2.
\end{equation}
Consequently, it follows from \eqref{eq:proplowerbound} of
Proposition~\ref{prop:lowerbound} that the capacity is lower bounded
by
\begin{equation}
  C(\SNR) \geq \frac{\tau}{\const{L}+\tau} \log\log\const{P}^{1/\tau}
  + \frac{\tau}{\const{L}+\tau} \Upsilon, \quad \const{P}>1.\label{eq:lowerfinite}
\end{equation}
Dividing by $\log\log\SNR$, and computing the limit as
$\SNR\to\infty$, yields
\begin{equation}
  \varliminf_{\SNR\to\infty} \frac{C(\SNR)}{\log\log\SNR} \geq \frac{\tau}{\const{L}+\tau},
\end{equation}
where we have used that for any fixed $\tau$
\begin{equation*}
  \lim_{\SNR\to\infty} \frac{\log\log\const{P}^{1/\tau}}{\log\log\SNR}
  = 1.
\end{equation*}
The lower bound on the capacity pre-loglog
\begin{equation}
  \Lambda \triangleq \varlimsup_{\SNR\to\infty} \frac{C(\SNR)}{\log\log\SNR}
  \geq \varliminf_{\SNR\to\infty} \frac{C(\SNR)}{\log\log\SNR} \geq 1
\end{equation}
follows then by letting $\tau$ tend to infinity. Together with the
upper bound $\Lambda\leq 1$, which was derived in
Section~\ref{sub:upper2}, this proves Theorem~\ref{thm:finite}.

\section{Conclusion}
\label{sec:summary}
We studied the high-SNR behavior of the capacity of noncoherent
multipath fading channels. We demonstrated that, depending on the
decay rate of the sequence $\{\alpha_{\ell}\}$, capacity may be
bounded or unbounded in the SNR. We further showed that if the number of paths is finite,
then at high SNR capacity grows double-logarithmically with the SNR, and the
capacity pre-loglog is irrespective of the
number of paths. The picture that emerges is as follows:
\begin{itemize}
\item If the sequence of variances $\{\alpha_{\ell}\}$ decays
  exponentially or slower, then capacity is bounded in the SNR.
\item If the sequence of variances $\{\alpha_{\ell}\}$ decays faster
  than exponentially, then capacity is unbounded in the SNR.
\item If the number of paths is finite, then the capacity pre-loglog
  is equal to $1$, irrespective of the number of paths.
\end{itemize}

The conclusions that can be drawn from these results are
twofold. First, multipath channels with an infinite
number of paths and multipath channels with a finite number of paths
have in general completely different capacity behaviors at high
SNR. Indeed, at high SNR, if the number of paths is finite, then
capacity grows double-logarithmically with the SNR, whereas if
the number of paths is infinite, then capacity may even be bounded in
the SNR. Thus, while for low or for moderate SNR it might be reasonable
to approximate a multipath channel with infinitely many paths by
a multipath channel with only a finite number paths, this is not
reasonable when the SNR tends to infinity. The number of paths that are
needed to approximate a multipath channel typically depends on the $\SNR$ and
may grow to infinity as the $\SNR$ tends to infinity.

Second, the above results indicate that the high-SNR behavior of
the capacity of multipath fading channels depends critically on the
assumed channel model. Thus when studying such channels at high SNR,
the channel modeling is crucial, as slight changes in the channel
model might lead to completely different capacity results.

\section*{Acknowledgment}
Fruitful discussions with Helmut B\"olcskei and Giuseppe Durisi are
gratefully acknowledged. The authors also wish to thank Olivier
Leveque and Nihar Jindal for their comments, which were the inspiration
for the proof of Proposition~\ref{prop:lowerbound}.

\appendix

\section{Appendix to Section~\ref{sub:upper1}}
\label{app:EPI}
To prove \eqref{eq:entropy1}, we lower bound
\begin{equation}
  \label{app:expr}
  h\Biggl(\sum_{\ell=0}^{k-1}
      H_k^{(\ell)}X_{k-\ell}+Z_k\Biggm|X_1^n=x_1^n,\vect{H}_1^{k-1}=\vect{h}_1^{k-1}\Biggr)
\end{equation}
for a given $\vect{h}_1^{k-1}$, and
average then the result over $\vect{H}_1^{k-1}$. Let $\set{H}_k$ denote the set
\begin{equation}
  \set{H}_k \triangleq \bigl\{H_k^{(\ell)}, \ell=0,\ldots,k-1:\alpha_{\ell}=0\bigr\}.
\end{equation}
We have
\begin{IEEEeqnarray}{lCl}
  \IEEEeqnarraymulticol{3}{l}{h\Biggl(\sum_{\ell=0}^{k-1}
  H_k^{(\ell)}X_{k-\ell}+Z_k\Biggm|X_1^n=x_1^n,\vect{H}_1^{k-1}=\vect{h}_1^{k-1}\Biggr)}
  \nonumber\\
  \qquad & \geq & h\Biggl(\sum_{\ell=0}^{k-1}
  H_k^{(\ell)}X_{k-\ell}+Z_k\Biggm|X_1^n=x_1^n,\vect{H}_1^{k-1}=\vect{h}_1^{k-1},\set{H}_k\Biggr)\nonumber\\
  & = & h\Biggl(\sum_{\ell\in\set{S}_k}
  H_k^{(\ell)}X_{k-\ell}+Z_k\Biggm|X_1^n=x_1^n,\vect{H}_1^{k-1}=\vect{h}_1^{k-1},\set{H}_k\Biggr)\nonumber\\
  & \geq &
  \log\Biggl(\sum_{\ell\in\set{S}_k}e^{h\Bigl(H_k^{(\ell)}X_{k-\ell}
  \Bigm| X_1^n=x_1^n,\bigl\{H_{k'}^{(\ell)}\bigr\}_{k'=1}^{k-1}=\bigl\{h_{k'}^{(\ell)}\bigr\}_{k'=1}^{k-1}\Bigr)}+e^{h(Z_k)}\Biggr),
  \label{app:beforeexp}
\end{IEEEeqnarray}
where $\set{S}_k$ is defined in \eqref{eq:S}. Here the first
inequality follows because conditioning cannot increase differential
entropy; the following equality follows because differential entropy
is invariant under deterministic translation
\cite[Thm.~9.6.3]{coverthomas91}, and because the terms where
$x_{k-\ell}=0$ do not contribute to the sum; and the last inequality
follows by the entropy power inequality
\cite[Thm.~16.6.3]{coverthomas91}, and because the processes
\begin{equation*}
  \bigl\{H_k^{(0)}, \,k\in\Naturals\bigr\}, \bigl\{H_k^{(1)},\,k\in\Naturals\bigr\},\ldots
\end{equation*}
are independent. (Note that, for a given $\vect{H}_1^{k-1}=\vect{h}_1^{k-1}$, the
conditional entropies on the RHS of \eqref{app:beforeexp} are possibly
infinite. However, by \eqref{eq:finiteentropy} this event is of zero
probability and is therefore immaterial to \eqref{app:beforeexp} when
averaged over $\vect{H}_{1}^{k-1}$.)

Since the processes of the path gains are independent and jointly
independent of $X_1^n$, we can compute the expectation of
\eqref{app:beforeexp} over $\vect{H}_1^{k-1}$ by averaging
\eqref{app:beforeexp} first over $(H_{1}^{(0)},\ldots,H_{k-1}^{(0)})$,
then averaging the result over $(H_{1}^{(1)},\ldots,H_{k-1}^{(1)})$,
and so on. To lower bound the individual expectations, we note that
the function
\begin{equation}
  f(x) = \log\bigl(e^{x}+\zeta\bigr), \quad x\in\Reals
\end{equation}
is convex for all $\zeta>0$. Thus, by setting for each
$\ell'=0,\ldots,k-1$
\begin{IEEEeqnarray}{lCl}
  \zeta_{\ell'} & = & \sum_{\substack{\ell\in\set{S}_k,\\\ell<\ell'}}e^{h\Bigl(H_k^{(\ell)}X_{k-\ell}
        \Bigm|
        X_1^n=x_1^n,\bigl\{H_{k'}^{(\ell)}\bigr\}_{k'=1}^{k-1}\Bigr)}\nonumber\\
        & & {} + \sum_{\substack{\ell\in\set{S}_k,\\\ell>\ell'}}e^{h\Bigl(H_k^{(\ell)}X_{k-\ell}
        \Bigm|
        X_1^n=x_1^n,\bigl\{H_{k'}^{(\ell)}\bigr\}_{k'=1}^{k-1}=\bigl\{h_{k'}^{(\ell)}\bigr\}_{k'=1}^{k-1}\Bigr)}+e^{h(Z_k)},
\end{IEEEeqnarray}
it follows from Jensen's inequality
\begin{IEEEeqnarray}{lCl}
  \IEEEeqnarraymulticol{3}{l}{\E[\bigl\{H_{k'}^{(\ell')}\bigr\}_{k'=1}^{k-1}]{\log\Biggl(\I{\ell'\in\set{S}_k}e^{h\Bigl(H_k^{(\ell')}X_{k-\ell'}
  \Bigm| X_1^n=x_1^n,\bigl\{H_{k'}^{(\ell')}\bigr\}_{k'=1}^{k-1}=\bigl\{h_{k'}^{(\ell)}\bigr\}_{k'=1}^{k-1}\Bigr)}+\zeta_{\ell'}\Biggr)}}\nonumber\\
  \quad \qquad & \geq & \log\Biggl(\I{\ell'\in\set{S}_k}e^{h\Bigl(H_k^{(\ell')}X_{k-\ell'}
  \Bigm|
  X_1^n=x_1^n,\bigl\{H_{k'}^{(\ell')}\bigr\}_{k'=1}^{k-1}\Bigr)}+\zeta_{\ell'}\Biggr),
  \qquad \ell'=0,\ldots,k-1,\IEEEeqnarraynumspace\label{app:individual}
\end{IEEEeqnarray}
where $\I{\cdot}$ denotes the indicator function, i.e.,
\begin{equation}
  \I{\textnormal{statement}} = \left\{\begin{array}{ll} 1 \quad &
  \textnormal{if statement is true}\\0 & \textnormal{if statement is false.}\end{array}\right.
\end{equation}
Averaging \eqref{app:beforeexp} over $\vect{H}_1^{k-1}$, and employing
\eqref{app:individual} to compute this average, yields thus
\begin{IEEEeqnarray}{lCl}
  \IEEEeqnarraymulticol{3}{l}{h\Biggl(\sum_{\ell=0}^{k-1}
  H_k^{(\ell)}X_{k-\ell}+Z_k\Biggm|X_1^n=x_1^n,\vect{H}_1^{k-1}\Biggr)}
  \nonumber\\
  \qquad \qquad \qquad \qquad \qquad & \geq & \log\Biggl(\sum_{\ell\in\set{S}_k}e^{h\Bigl(H_k^{(\ell)}X_{k-\ell}
  \Bigm| X_1^n=x_1^n,\bigl\{H_{k'}^{(\ell)}\bigr\}_{k'=1}^{k-1}\Bigr)}+e^{h(Z_k)}\Biggr).
\end{IEEEeqnarray}
This proves the lower bound \eqref{eq:entropy1}.


\end{document}